\makeatletter
\let\@twosidetrue\@twosidefalse
\let\@mparswitchtrue\@mparswitchfalse
\makeatother

\documentclass{llncs}
\usepackage{amsfonts}
\usepackage{amssymb}
\usepackage{amsmath}
\usepackage{tabto}
\usepackage{pdflscape}
\usepackage{multirow}
\usepackage{setspace}
\usepackage{textcomp}
\usepackage{algorithm}
\usepackage{algorithmicx}
\usepackage{algpseudocode}
\usepackage{enumitem}
\usepackage{tikz,complexity,float}
\usepackage{xcolor,colortbl,bm,graphicx,adjustbox}

\usetikzlibrary{decorations.pathmorphing}
\tikzset{snake it/.style={decorate, decoration=snake}}

\definecolor{MyColor}{RGB}{197,0,205}

\newtheorem{pr}{Problem}
\newcommand{\inp}{\textsf{Input: }} 

\newcommand{\ques}{\textsf{Question: }}
\usepackage[numbers]{natbib}

\title{The stable marriage problem with ties and restricted edges
\thanks{The authors were supported by the Cooperation of Excellences Grant (KEP-6/2018), the Hungarian Academy of Sciences under its Momentum Programme (LP2016-3/2016), its J\'anos Bolyai Research Fellowship and OTKA grant K128611, the DFG Research Training Group 2434 ``Facets of Complexity'', and COST Action CA16228 European Network for Game Theory.}} 
\author{\'{A}gnes Cseh\inst{1} \and Klaus Heeger\inst{2}}
\institute{
Institute of Economics, Centre for Economic and Regional Studies, Hungarian Academy of Sciences, Hungary, \texttt{cseh.agnes@krtk.mta.hu}
\and
Algorithmics and Computational Complexity, Faculty~IV, TU Berlin, Germany, \texttt{heeger@tu-berlin.de}
}

\begin{document}
\maketitle

\begin{abstract}
    In the stable marriage problem, a set of men and a set of women are given, each of whom has a strictly ordered preference list over the acceptable agents in the opposite class. A matching is called stable if it is not blocked by any pair of agents, who mutually prefer each other to their respective partner. Ties in the preferences allow for three different definitions for a stable matching: weak, strong and super-stability. Besides this, acceptable pairs in the instance can be restricted in their ability of blocking a matching or being part of it, which again generates three categories of restrictions on acceptable pairs. Forced pairs must be in a stable matching, forbidden pairs must not appear in it, and lastly, free pairs cannot block any matching.
    
    Our computational complexity study targets the existence of a stable solution for each of the three stability definitions, in the presence of each of the three types of restricted pairs. We solve all cases that were still open. As a byproduct, we also derive that the maximum size weakly stable matching problem is hard even in very dense graphs, which may be of independent interest.
\end{abstract}

\section{Introduction}

In the classical \emph{stable marriage problem} ({\sc sm})~\cite{GS62}, a bipar\-tite graph is given, where one side symbolizes a set of men~$U$, while the other side symbolizes a set of women~$W$. Man $u$ and woman $w$ are connected by edge $uw$ if they find one another mutually acceptable. In the most basic setting, each participant provides a strictly ordered preference list of the acceptable agents of the opposite gender. An edge $uw$ \emph{blocks} matching $M$ if it is not in $M$, but each of $u$ and $w$ is either unmatched or prefers the other to their respective partner in~$M$. A \emph{stable matching} is a matching not blocked by any edge. From the seminal paper of Gale and Shapley~\cite{GS62}, we know that the existence of such a stable solution is guaranteed and one can be found in linear time.

Several real-world applications~\cite{Bir17} require a relaxation of the strict order to weak order, or, in other words, preference lists with ties, leading to the \emph{stable marriage problem with ties} ({\sc smt})~\cite{Irv94,IMS03,Man02}.
When ties occur, the definition of a blocking edge needs to be revisited. In the literature, three intuitive definitions are used, namely weakly, strongly and super-stable matchings~\cite{Irv94}. According to weak stability, a matching is \emph{weakly blocked} by an edge $uw$ if agents $u$ and $w$ both strictly prefer one another to their partners in the matching. A \emph{strongly blocking} edge is preferred strictly by one end vertex, whereas it is not strictly worse than the matching edge at the other end vertex. A \emph{super-blocking edge} is at least as good as the matching edge for both end vertices in the super-stable case. Super-stable matchings are strongly stable and strongly stable matchings are weakly stable by definition, because weakly blocking edges are strongly blocking, and strongly blocking edges are super-blocking at the same time.

Another classical direction of research is to distinguish some of the edges based on their ability to be part of or to block a matching. Table~\ref{ta:restricted_edges} provides a structured overview of the three sorts of restricted edges that have been defined in earlier papers~\cite{Knu76,DFFS03,FIM07,AIKMP13,Kwa15,CM16}. The mechanism designer has the right to specify three sets of restricted edges: \emph{forced} edges must be in the output matching, \emph{forbidden} edges must not appear in it, and finally, \emph{free} edges cannot block the matching, no matter the preference ordering.

\begin{table*}[]
	\centering
      \begin{adjustbox}{max width=1\textwidth}
		\begin{tabular}{|l| c c c|}
		\hline
			& $uw$ must be in $M$ & $uw$ can be in $M$ & $uw$ must not be in $M$ \\ \hline
			$uw$ can block $M$  & \textbf{forced} & unrestricted& \textbf{forbidden} \\ 
			$uw$ cannot block $M$ & \textbf{forced} & \textbf{free} & irrelevant \\
		\hline
		\end{tabular}
        \end{adjustbox}
    \caption{The three types of restricted edges are marked with bold letters. The columns tells edge $uw$'s role regarding being in a matching, while the rows split cases based on $uw$'s ability to block a matching.}
    \label{ta:restricted_edges}
\end{table*}

In this paper, we combine weakly ordered lists and restricted edges, and determine the computational complexity of finding a stable matching in all cases not solved yet.

\subsection{Literature review}
\label{sec:literature}

We first focus on the known results for the {\sc smt} problem without restricted edges, and then switch to the {\sc sm} problem with edge restrictions. Finally, we list all progress up to our paper in {\sc smt} with restricted edges. 

\paragraph{Ties.} If all edges are unrestricted, a weakly stable matching always exists, because generating any linear extension to each preference list results in a classical \textsc{sm} instance, which admits a solution~\cite{GS62}. This solution remains stable in the original instance as well. On the other hand, strong and super-stable matchings are not guaranteed to exist. However, there are polynomial-time algorithms to output a strongly/su\-per-stable matching or a proof for its nonexistence~\cite{Irv94,Man13}.

\paragraph{Restricted edges.} Dias et al.~\cite{DFFS03} showed that the problem of finding a stable matching in a {\sc sm} instance with forced and forbidden edges or reporting that none exists is solvable in $O(m)$ time, where $m$ is the number of edges in the instance. Approximation algorithms for instances not admitting any stable matching including all forced and avoiding all forbidden edges  
were studied in~\cite{CM16}. The existence of free edges can only enlarge the set of stable solutions, thus a stable matching with free edges always exists. However, in the presence of free edges, a maximum-cardinality stable matching is $\NP$-hard to find~\cite{AIKMP13}. Kwanashie~\cite[Sections 4 and 5]{Kwa15} performed an exhaustive study on various stable matching problems with free edges.  The term ``stable with free edges''~\cite{CF09,FIM11} is equivalent to ``socially stable''~\cite{AIKMP13,Kwa15}.

\paragraph{Ties and restricted edges.} Table~\ref{ta:results} illustrates the known and our new results on problems that arise when ties and restricted edges are combined in an instance. Weakly stable matchings in the presence of forbidden edges were studied by Scott~\cite{Sco05}, where the author shows that deciding whether a matching exists avoiding the set of forbidden edges is $\NP$-complete. A similar hardness result was derived by Manlove et al.~\cite{MIIMM02} for the case of forced edges, even if the instance has a single forced edge. Forced and forbidden edges in super-stable matchings were studied by Fleiner et al.~\cite{FIM07}, who gave a polynomial-time algorithm to decide whether a stable solution exists. Strong stability in the presence of forced and forbidden edges is covered by Kunysz~\cite{Kun18}, who gave a polynomial-time algorithm for the weighted strongly stable matching problem with non-negative edge weights. Since strongly stable matchings are always of the same cardinality~\cite{Man99,IMS03}, a stable solution or a proof for its nonexistence can be found via setting the edge weights to 0 for forbidden edges, 2 for forced edges, and 1 for unrestricted edges.

\subsection{Our contributions}

In Section~\ref{se:forbidden} we prove a stronger result than the hardness proof in~\cite{Sco05} delivers: we show that finding a weakly stable matching in the presence of forbidden edges is $\NP$-complete even if the instance has a single forbidden edge.

As a byproduct, we gain insight into the well-known maximum size weakly stable matching (without any edge restriction). This problem is known to be $\NP$-complete~\cite{IMMM99,MIIMM02}, even if preference lists are of length at most three~\cite{IMO09,MM10}. On the other hand, if the graph is complete, a complete weakly stable matching is guaranteed to exist. It turns out that this completeness is absolutely crucial to keep the problem tract\-able: as we show here, if the graph is a complete bipartite graph missing exactly one edge, then deciding whether a perfect weakly stable matching exists is $\NP$-complete.

We turn to the problem of free edges under strong and super-stability in Section~\ref{se:free}. We show that deciding whether a strongly/super-stable exists when free edges occur in the instance is $\NP$-complete. This hardness is in sharp contrast to the polynomial-time algorithms for the weighted strongly/ super-stable matching problems.
Afterwards, we show that deciding the existence of a strongly or super-stable matching in an instance with free edges is fixed-parameter tractable parameterized by the number of free edges. 

\begin{table*}[]
	\centering
      \begin{adjustbox}{max width=1\textwidth}
		\begin{tabular}{|l|c c c|}
		\hline
			Existence& weak & strong & super \\ \hline
            forbidden & $\NP$-complete~\cite{Sco05} \textbf{{\color{gray!70} even if $|P|=1$} }& $O(nm)$~\cite{Kun18}&$O(m)$~\cite{FIM07}\\
            forced & $\NP$-complete even if $|Q|=1$~\cite{MIIMM02} & $O(nm)$~\cite{Kun18}&$O(m)$~\cite{FIM07}\\
            free & always exists & \textbf{{\color{gray!70}$\NP$-complete}}&\textbf{{\color{gray!70}$\NP$-complete}}\\
		\hline
		\end{tabular}
        \end{adjustbox}
    \caption{Previous and our results summarized in a table. The contribution of this paper is marked by bold gray font. The instance has $n$ vertices, $m$ edges, $|P|$ forbidden edges, and $|Q|$~forced edges.}
		\label{ta:results}
\end{table*}

\section{Preliminaries}
\label{se:prel}

The input of the stable marriage problem with ties consists of a bipartite graph $G = (U \cup W, E)$ and for each $v \in U \cup W$, a weakly ordered preference list $O_v$ of the edges incident to~$v$. We denote the number of vertices in $G$ by $n$, while $m$ stands for the number of edges. An edge connecting vertices $u$ and~$w$ is denoted by~$uw$. We say that the preference lists in an instance are derived from a \emph{master list} if there is a weak order $O$ of $U \cup W$ so that each $O_v$ where $v \in U \cup W$ can be obtained by deleting entries from~$O$.

The set of restricted edges consists of the set of \emph{forbidden edges} $P$, the set of \emph{forced edges} $Q$, and the set of \emph{free edges}~$F$. These three sets are disjoint. 
\begin{definition}
A matching $M$ is \emph{weakly/strongly/super-stable with restricted edges $P,Q$, and $F$}, if $M \cap P = \emptyset$, $Q \subseteq M$, and the set of edges blocking $M$ in a weakly/strongly/super sense is a subset of~$F$.
\end{definition}

\section{Weak stability}
\label{se:forbidden}

In Theorem~\ref{thm:NP-hardForbidden} we present a hardness proof for the weakly stable matching problem with a single forbidden edge, even if this edge is ranked last by both end vertices. The hardness of the maximum-cardinality weakly stable matching problem in dense graphs (Theorem~\ref{th:dense}) follows easily from this result.

\begin{pr}\textsc{smt-forbidden-1} \\
	\inp A complete bipartite graph $G=(U \cup W,E)$, a forbidden edge $P = \left\{uw \right\}$ and preference lists with ties.\\
	\ques Does there exist a weakly stable matching $M$ so that $uw \notin M$?
\end{pr}

\begin{theorem}\label{thm:NP-hardForbidden}
    \textsc{smt-forbidden-1} is $\NP$-complete, even if all ties are of length two, they appear only on one side of the bipartition and at the beginning of the complete preference lists, and the forbidden edge is ranked last by both its end vertices.
\end{theorem}

\begin{proof}
 \textsc{smt-forbidden-1} is clearly in $\NP$, as any matching can be checked for stability in linear time.

 We reduce from the \textsc{perfect-smti} problem defined below, which is known to be $\NP$-complete even if all ties are of length two, and appear on one side of the bipartition and at the beginning of the preference lists, as shown by Manlove et al.~\cite{MIIMM02}.
 
 \begin{pr}\textsc{perfect-smti} \\
	\inp An incomplete bipartite graph $G=(U \cup W,E)$, and preference lists with ties.\\
	\ques Does there exist a perfect weakly stable matching~$M$?
 \end{pr}
 
 \paragraph{Construction.} To each instance $\mathcal{I}$ of \textsc{perfect-smti}, we construct an instance $\mathcal{I'}$ of \textsc{smt-forbidden-1}.
 
 Let $G=(U \cup W,E)$ be the underlying graph in instance~$\mathcal{I}$. When constructing $G'$ for $\mathcal{I'}$, we add two men $u_1$ and $u_2$ to $U$, and two women $w_1$ and $w_2$ to~$W$. On vertex classes $U' = U \cup \{u_1, u_2\}$ and $W' = W \cup \{w_1, w_2\}$, $G'$ will be a complete bipartite graph. As the list below shows, we start with the original edge set $E(G)$ in stage~0, and then add the remaining edges in four further stages. 
 An example for the built graph is shown in Figure~\ref{fig:max-SMTI-dense-reduction}.
 
 \begin{enumerate}[start=0]
    \item $E(G)$\\
    We keep the edges in~$E(G)$ and also preserve the vertices' rankings on them. These edges are solid black in Figure~\ref{fig:max-SMTI-dense-reduction}.
    \item $(U\times \{w_1\})\cup (\{u_1\} \times W)$\\
    We first connect $u_1$ to all women in $W$, and $w_1$ to all men in~$U$.
    Man $u_1$ (woman $w_1$) ranks the women from~$W$ (men from $U$) in an arbitrary order.
    Each $u \in U$ ($w\in W$) ranks $w_1$ ($u_1$) after all their edges in~$E(G)$. These edges are loosely dashed green in Figure~\ref{fig:max-SMTI-dense-reduction}.
 
    \item $(U\times W) \setminus E(G)$\\
    Now we add for each pair $(u, w)\in U\times W$ with $u w\notin E(G)$ the edge $u w$, where $u$ ($w$) ranks $w$ ($u$) even after $w_1$ ($u_1$). These edges are densely dashed blue in Figure~\ref{fig:max-SMTI-dense-reduction}.
 
    \item $\bigl[(U\cup\{u_1\}) \times \{w_2\}\bigr]\cup \bigl[ \{u_2\} \times (W\cup \{w_1\})\bigr]$\\
    Man $u_2$ is connected to all women from $W\cup \{w_1\}$, and ranks all these women in an arbitrary order.
    The women from $W\cup \{w_1\}$ rank $u_2$ worse than any already added edge.
    Similarly, $w_2$ is connected to all men from $M\cup \{u_1\}$, and ranks all these men in an arbitrary order.
    The men from $M\cup \{u_1\}$ rank $w_2$ worse than any already added edge. These edges are dotted red in Figure~\ref{fig:max-SMTI-dense-reduction}.
 
     \item ${u_1 w_1}$ and ${u_2 w_2}$\\
     Finally, we add the edges $u_1 w_1$ and $u_2 w_2$, which are ranked last by both of their end vertices.
     Edge $u_2 w_2$ is the only forbidden edge and it is the violet zigzag edge in Figure~\ref{fig:max-SMTI-dense-reduction}, while $u_2w_2 $ is wavy gray.
 \end{enumerate}
  
 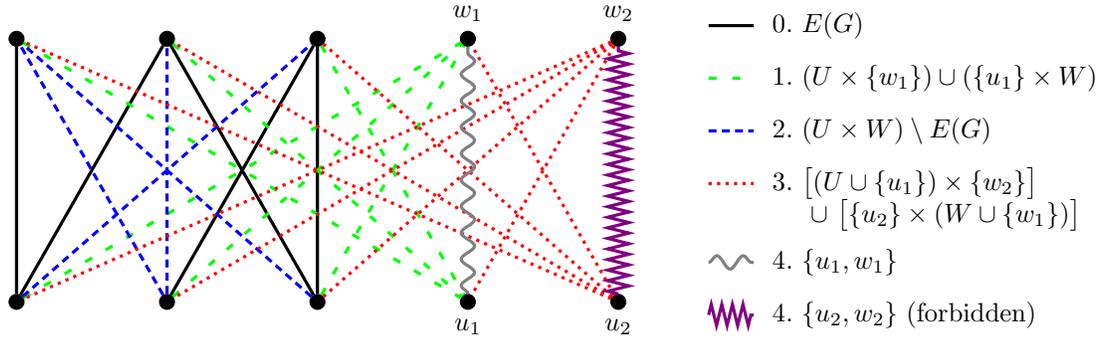
\begin{figure*}
     \centering
     \begin{tikzpicture}[yscale = 3.5, xscale=2.0]
       \tikzstyle{vertex}=[draw, circle, fill, inner sep = 2pt]
       
       \node[vertex, label=270:$u_1$] (m1) at (3, 0) {};
       \node[vertex, label=270:$u_2$] (m2) at (4, 0) {};
       \node[vertex] (m3) at (2, 0) {};
       \node[vertex] (m4) at (1, 0) {};
       \node[vertex] (m5) at (0, 0) {};
       
       \node[vertex, label=90:$w_1$] (w1) at (3, 1) {};
       \node[vertex, label=90:$w_2$] (w2) at (4, 1) {};
       \node[vertex] (w3) at (2, 1) {};
       \node[vertex] (w4) at (1, 1) {};
       \node[vertex] (w5) at (0, 1) {};
       
       \draw (m3) edge[very thick] (w3);
       \draw (m3) edge[very thick] (w4);
       \draw (m4) edge[very thick] (w3);
       \draw (m5) edge[very thick] (w5);
       \draw (m5) edge[very thick] (w4);
       
       \draw (m3) edge[green, loosely dashed, very thick] (w1);
       \draw (m4) edge[green, loosely dashed, very thick] (w1);
       \draw (m5) edge[green, loosely dashed, very thick] (w1);
       \draw (w3) edge[green, loosely dashed, very thick] (m1);
       \draw (w4) edge[green, loosely dashed, very thick] (m1);
       \draw (w5) edge[green, loosely dashed, very thick] (m1);
       
       \draw (m3) edge[blue, densely dashed, very thick] (w5);
       \draw (m4) edge[blue, densely dashed, very thick] (w4);
       \draw (m4) edge[blue, densely dashed, very thick] (w5);
       \draw (m5) edge[blue, densely dashed, very thick] (w3);
       
       \draw (m2) edge[red, dotted, very thick] (w1);
       \draw (m2) edge[red, dotted, very thick] (w3);
       \draw (m2) edge[red, dotted, very thick] (w4);
       \draw (m2) edge[red, dotted, very thick] (w5);
       \draw (w2) edge[red, dotted, very thick] (m1);
       \draw (w2) edge[red, dotted, very thick] (m3);
       \draw (w2) edge[red, dotted, very thick] (m4);
       \draw (w2) edge[red, dotted, very thick] (m5);

       \draw (m1) edge[gray, snake it, very thick] (w1);
       
       \draw (m2) edge[very thick, violet, decoration = {zigzag, segment length=4, amplitude = 4}, decorate] (w2);
       
       \begin{scope}[yshift=0.05cm, xshift=-0.1cm]
           \draw (4.7,1)  edge[very thick] (5, 1);
           \node[label=0:0. $E(G)$] (l1) at (5, 1) {};
           \draw (4.7, 0.8) edge[very thick, green, loosely dashed] (5, 0.8);
           \node[label=0:1. $(U\times \{w_1\})\cup (\{u_1\} \times W)$] (l2) at (5, 0.8) {};
           \draw (4.7,0.6) edge[very thick, blue, densely dashed] (5, 0.6);
           \node[label=0:2. $(U\times W) \setminus E(G)$] (l3) at (5, 0.6) {};
           \draw (4.7,0.4) edge[very thick, red, dotted] (5, 0.4);
           \node[label=0:3. $\bigl{[} ( U\cup \{u_1\}) \times \{w_2\}\bigr{]}$] (l4) at (5, 0.4) {};
           \node[label=0:{\color{white} 3. }$ \ \cup \ \bigl{[}  \{u_2\} \times (W\cup \{w_1\})\bigr{]}$] (l45) at (5, 0.28) {};
           \draw (4.7,0.1) edge[very thick, gray, snake it] (5, 0.1);
           \node[label=0:4. ${\{u_1, w_1\}}$] (l5) at (5, 0.1) {};
           \draw (4.7, -0.1) edge[very thick, violet, decoration = {zigzag, segment length=4, amplitude = 4}, decorate] (5, -0.1);
           \node[label=0:4. ${\{u_2, w_2\}}$ (forbidden)] (l6) at (5, -0.1) {};
       \end{scope}
     \end{tikzpicture}
     \caption{An example for the reduction. The legend on the right side lists the five groups of edges in the preference order at all vertices. The edges from the \textsc{perfect-smti} instance (drawn in solid black) keep their ranks.
     Every vertex ranks solid black edges best, then loosely dashed green edges, then densely dashed blue edges, then dotted red edges, then the wavy gray edge $\{u_1, w_1\}$ and the forbidden violet zigzag edge $\{u_w, w_2\}$.}
     \label{fig:max-SMTI-dense-reduction}
 \end{figure*}{}
 
 \bfseries Claim: \mdseries $\mathcal{I}$ admits a perfect stable matching if and only if $\mathcal{I}'$ admits a stable matching not containing~$u_2 w_2$.

 ($\Rightarrow$) Let $M$ be a perfect stable matching in~$\mathcal{I}$. 
 We construct $M'$ as $M \cup \{u_1 w_2\} \cup \{u_2 w_1\}$.
 Clearly, $M'$ is a matching not containing the forbidden edge $u_2 w_2$, so it only remains to show that $M'$ is stable. We do this by case distinction on a possible blocking edge.
 
\begin{enumerate}[start=0]
    \item $E(G)$\\
    Since $M$ does not admit a blocking edge in $\mathcal{I}$, no edge from the original $E(G)$ can block $M'$ in $\mathcal{I}'$.
    \item $(U\times \{w_1\}) \cup (\{u_1\} \times W)$\\
    All vertices in $U \cup W$ rank these edges lower than their edges in~$M'$.
    \item $(U\times W) \setminus E(G)$\\
    Edges in this set cannot block $M'$ because they are ranked worse than edges in $M'$ by both of their end vertices.
    \item $\bigl[ (U \cup \{u_1\}) \times \{w_2\}\bigr] \cup \bigl[ ( W\cup \{w_1\}) \times \{u_2\}\bigr]$\\
    Vertices in $U \cup W$ prefer their edge in $M'$ to all edges in this set. Since they are in $M'$, $u_1 w_2$ and $u_2 w_1$ also cannot block~$M'$.
    \item ${u_1 w_1}$ and ${u_2 w_2}$\\
    These two edges are strictly worse than $u_1 w_2 \in M'$ and $u_2 w_1\in M'$ at all four end vertices.
\end{enumerate}
     
 ($\Leftarrow$) Let $M'$ be a stable matching in~$\mathcal{I}'$ and $u_2 w_2 \notin M'$.
 Since $G'$ is a complete bipartite graph with the same number of vertices on both sides, $M'$ is a perfect matching. In particular, $u_2$ and $w_2$ are matched by $M'$, say to~$w$ and~$u$, respectively.
 Since $M'$ does not contain the forbidden edge~$u_2 w_2$, we have that $u\neq u_2$ and $w\neq w_2$.
 Then we have $w= w_1$ and $u = u_1$, as $u w$ blocks $M'$ otherwise.
 
 If $M'$ contains an edge $u w\notin E(G)$ with $u\in U$ and $w\in W$, then this implies that $u w_1$ is a blocking edge.
 Thus, $M:= M'\setminus \{u_1 w_2, u_2 w_1\}\subseteq E(G)$, i.e.\ it is a perfect matching in~$G$.
 This $M$ is also stable, as any blocking edge in $G$ immediately implies a blocking edge for $M'$, which contradicts our assumption on $M'$ being a stable matching. 
\end{proof}

As a byproduct, 
we get that \textsc{max-smti-dense}, the problem of deciding whether an almost complete bipartite graph admits a perfect weakly stable matching, is also $\NP$-complete.

\begin{pr}\textsc{max-smti-dense} \\
	\inp A bipartite graph $G=(U \cup W,E)$, where $E(G) = K_{n, n} \setminus \{e\}$ and preference lists with ties.\\
	\ques Does there exist a perfect weakly stable matching~$M$?
\end{pr}

\begin{theorem}
\label{th:dense}
 \textsc{max-smti-dense} is $\NP$-complete, even if all ties are of length two, are on one side of the bipartition, and appear at the beginning of the preference lists.
\end{theorem}

\begin{proof}
 \textsc{max-smti-dense} is in $\NP$, as a matching can be checked for stability in linear time.
 
 We reduce from \textsc{smt-forbidden-1}.
 By Theorem~\ref{thm:NP-hardForbidden}, this problem is even $\NP$-complete if the forbidden edge $u w$ is at the end of the preference lists of $u$ and~$w$.
 For each such instance $\mathcal{I}$ of \textsc{smt-forbidden-1}, we construct an instance $\mathcal{I'}$ of \textsc{max-smti-dense} by deleting the forbidden edge $u w$.
 
 \bfseries Claim: \mdseries
 The instance $\mathcal{I}$ admits a stable matching if and only if $\mathcal{I'}$ admits a perfect stable matching.
 
 ($\Rightarrow$) Let $M$ be a stable matching for~$\mathcal{I}$.
 As \textsc{smt-forbid\-den-1} gets a complete bipartite graph as an input, $M$ is a perfect matching.
 Since $M$ does not contain the edge $u w$, it is also a matching in $\mathcal{I'}$.
 Moreover, $M$ is stable there, because the transformation only removed a possible blocking edge and added none of these.
 
 ($\Leftarrow$) Let $M'$ be a perfect stable matching in~$\mathcal{I'}$.
 Since~$u w$ is at the end of the preference lists of $u$ and $w$, and $M'$ is perfect, $u w$ cannot block~$M'$.
 Thus, $M'$ is stable in~$\mathcal{I}$.
\end{proof}

Having shown a hardness result for the existence of a weakly stable matching even in very restricted instances with a single forbidden edge in Theorem~\ref{thm:NP-hardForbidden}, we now turn our attention to strongly and super-stable matchings.

\section{Strong and super-stability}
\label{se:free}

As already mentioned in Section~\ref{sec:literature}, strongly and super-stable matchings can be found in polynomial time if forced and forbidden edges both occur in the instance~\cite{FIM07,Kun18}. Thus we consider the case of free edges, and in Theorem~\ref{thm:NP-hardFree} and Proposition~\ref{thm:free_complete} we show hardness for the strong and super-stable matching problems in instances with free edges. The same construction suits both cases. Then, in Proposition~\ref{thm:par_free} we remark that both problems are fixed-parameter tractable with $|F|$ as the parameter.

\begin{pr}\textsc{ssmti-free} \\
	\inp A bipartite graph $G=(U \cup W,E)$, a set $F \subseteq E$ of free edges, and preference lists with ties.\\
	\ques Does there exist a matching $M$ so that $uw \in F$ for all $uw \in E$ that blocks $M$ in the strongly/super-stable sense?
\end{pr}

In \textsc{ssmti-free}, we define two problem variants simultaneously, because all our upcoming proofs are identical for both of these problems. For the super-stable marriage problem with ties and free edges, all super-blocking edges must be in~$F$, while for the strongly stable marriage problem with ties and free edges, it is sufficient if a subset of these, the strongly blocking edges, are in~$F$.


\begin{theorem}\label{thm:NP-hardFree}
    \textsc{ssmti-free} is $\NP$-complete even in graphs with maximum degree four, and if preference lists of women are derived from a master list.
\end{theorem}

\begin{proof}
 \textsc{ssmti-free} is clearly in $\NP$ because the set of edges blocking a matching can be determined in linear time.

 We reduce from the \textsc{1-in-3 positive 3-sat} problem, defined below, which is known to be $\NP$-complete~\cite{Sch78,GJ79,PSSW14}.
  \begin{pr}\textsc{1-in-3 positive 3-sat} \\
	\inp A 3-SAT formula, in which no literal is negated and every variable occurs in at most three clauses.\\
	\ques Does there exist a satisfying truth assignment that sets exactly one literal in each clause to be true?
\end{pr}

\paragraph{Construction.} To each instance $\mathcal{I}$ of \textsc{1-in-3 positive 3-sat}, we construct an instance $\mathcal{I}'$ of \textsc{ssmti-free}.

Let $x_1, \ldots, x_n$ be the variables and $C_1, \dots, C_m$ be the claus\-es of the \textsc{1-in-3 positive 3-sat} instance $\mathcal{I}$. For each clause $C_i$, we add a clause gadget consisting of three vertices $a_i$, $b_i$, and~$c_i$, where~$b_i$ is connected to $a_i$ and $c_i$, as shown in Figure~\ref{fclausegadget}. While vertices~$a_i$ and $b_i$ do not have any further edge,~$c_i$ will be incident to three \emph{interconnecting} edges leading to variable gadgets. Vertex $b_i$ is ranked first by $a_i$ and last by~$c_i$, and these two vertices are placed in a tie by~$b_i$.
 
  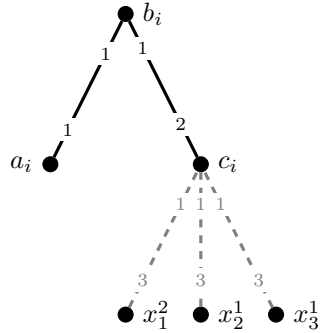
\begin{figure}
   \begin{center}
     \begin{tikzpicture}
       \tikzstyle{vertex}=[draw, circle, fill, inner sep = 2pt]
       
       \node[vertex, label=180:$a_{i}$] (ai) at (3, 0) {};
       \node[vertex, label=0:$b_{i}$] (bi) at (4, 2) {};
       \node[vertex, label=0:$c_{i}$] (ci) at (5, 0) {};

       \draw (ai) edge[very thick] node[pos=0.2, fill=white, inner sep=2pt] {\scriptsize $1$}  node[pos=0.76, fill=white, inner sep=2pt] {\scriptsize $1$} (bi);
        
       \draw (bi) edge[very thick] node[pos=0.2, fill=white, inner sep=2pt] {\scriptsize $1$}  node[pos=0.76, fill=white, inner sep=2pt] {\scriptsize $2$} (ci);
       
       \node[vertex, label=0:$x_1^2$] (x1) at (4, -2) {};
       \node[vertex, label=0:$x_2^1$] (x2) at (5, -2) {};
       \node[vertex, label=0:$x_3^1$] (x3) at (6, -2) {};
       
       \draw (x1) edge[dashed, gray, very thick] node[pos=0.2, fill=white, inner sep=2pt] {\scriptsize $3$}  node[pos=0.76, fill=white, inner sep=2pt] {\scriptsize $1$} (ci);
       \draw (x2) edge[dashed, gray, very thick] node[pos=0.2, fill=white, inner sep=2pt] {\scriptsize $3$}  node[pos=0.76, fill=white, inner sep=2pt] {\scriptsize $1$} (ci); 
       \draw (x3) edge[dashed, gray, very thick] node[pos=0.2, fill=white, inner sep=2pt] {\scriptsize $3$}  node[pos=0.76, fill=white, inner sep=2pt] {\scriptsize $1$} (ci);  
     \end{tikzpicture}

   \end{center}
   \caption{An example of a clause gadget for the clause $C_i$, containing the variables $x_1$, $x_4$, and $x_5$.
   The interconnecting edges are dashed and gray.}\label{fclausegadget}
 \end{figure}
 
 For each variable $x_i$, occurring in the three clauses $C_{i_1}$,~$C_{i_2}$, and $C_{i_3}$, we add a variable gadget with nine vertices $y^j_i$, $z^j_i$, and~$w_i^j$ for $j\in [3]$, as indicated in Figure~\ref{fvariablegadget}.
 Each vertex $z^j_i$ is connected only to $y^j_i$ by a free edge, and these are the only free edges in our construction.
 For each $(\ell, j)\in [3]^2$, we add an edge $w^\ell_i y^j_i$, which is ranked second (after $z^j_i$) by~$y^j_i$.
 The vertex $w^\ell_i$ ranks this edge at position one if $\ell=j$ and else at position two.
 Finally, we connect the vertex $w^\ell_i$ to the vertex $c_{i_\ell}$ by an interconnecting edge, ranked at position one by $c_{i_\ell}$ and position three by~$w^\ell_i$.

 The resulting instance is bipartite: $U= \{z_i^j, w_i^j, b_i\}$ is the set of men and $W = \{y_i^j, c_i, a_i\}$ is the set of women.
 One easily sees that the maximum degree in our reduction is four.
 
 Note that the preference lists of the women in the \textsc{ssmti-free} instance are derived from a master list. The master list for the women $W = \{y_i^j, c_i, a_i\}$ is the following.
 At the top are all vertices of the form $\{z_i^j\}$ in a single tie, followed by all vertices of the form $\{w_i^j\}$ in a single tie, and finally, all other vertices ($\{b_i\}$) at the bottom of the preference list.
 
  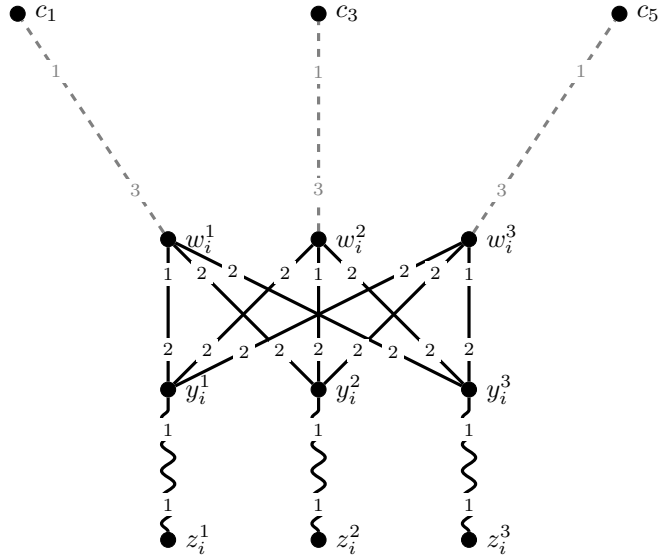
\begin{figure}
   \begin{center}
     \begin{tikzpicture}[scale=2]
       \tikzstyle{vertex}=[draw, circle, fill, inner sep = 2pt]
       
       \node[vertex, label=0:$y^1_{i}$] (y1) at (0, 1) {};
       \node[vertex, label=0:$z^1_{i}$] (z1) at (0, 0) {};
       \node[vertex, label=0:$w^1_{i}$] (w1) at (0, 2) {};
       
       \node[vertex, label=0:$y^2_{i}$] (y2) at (1, 1) {};
       \node[vertex, label=0:$z^2_{i}$] (z2) at (1, 0) {};
       \node[vertex, label=0:$w^2_{i}$] (w2) at (1, 2) {};
       
       \node[vertex, label=0:$y^3_{i}$] (y3) at (2, 1) {};
       \node[vertex, label=0:$z^3_{i}$] (z3) at (2, 0) {};
       \node[vertex, label=0:$w^3_{i}$] (w3) at (2, 2) {};

       \draw (z1) edge[snake it, very thick] node[pos=0.2, fill=white, inner sep=2pt] {\scriptsize $1$}  node[pos=0.76, fill=white, inner sep=2pt] {\scriptsize $1$} (y1);
       \draw (z2) edge[snake it, very thick] node[pos=0.2, fill=white, inner sep=2pt] {\scriptsize $1$}  node[pos=0.76, fill=white, inner sep=2pt] {\scriptsize $1$} (y2);
       \draw (z3) edge[snake it, very thick] node[pos=0.2, fill=white, inner sep=2pt] {\scriptsize $1$}  node[pos=0.76, fill=white, inner sep=2pt] {\scriptsize $1$} (y3);
        
       \draw (w1) edge[very thick] node[pos=0.2, fill=white, inner sep=2pt] {\scriptsize $1$}  node[pos=0.76, fill=white, inner sep=2pt] {\scriptsize $2$} (y1); 
       \draw (w1) edge[very thick] node[pos=0.2, fill=white, inner sep=2pt] {\scriptsize $2$}  node[pos=0.76, fill=white, inner sep=2pt] {\scriptsize $2$} (y2); 
       \draw (w1) edge[very thick] node[pos=0.2, fill=white, inner sep=2pt] {\scriptsize $2$}  node[pos=0.76, fill=white, inner sep=2pt] {\scriptsize $2$} (y3);
        
       \draw (w2) edge[very thick] node[pos=0.2, fill=white, inner sep=2pt] {\scriptsize $2$}  node[pos=0.76, fill=white, inner sep=2pt] {\scriptsize $2$} (y1); 
       \draw (w2) edge[very thick] node[pos=0.2, fill=white, inner sep=2pt] {\scriptsize $1$}  node[pos=0.76, fill=white, inner sep=2pt] {\scriptsize $2$} (y2); 
       \draw (w2) edge[very thick] node[pos=0.2, fill=white, inner sep=2pt] {\scriptsize $2$}  node[pos=0.76, fill=white, inner sep=2pt] {\scriptsize $2$} (y3);
        
       \draw (w3) edge[very thick] node[pos=0.2, fill=white, inner sep=2pt] {\scriptsize $2$}  node[pos=0.76, fill=white, inner sep=2pt] {\scriptsize $2$} (y1); 
       \draw (w3) edge[very thick] node[pos=0.2, fill=white, inner sep=2pt] {\scriptsize $2$}  node[pos=0.76, fill=white, inner sep=2pt] {\scriptsize $2$} (y2); 
       \draw (w3) edge[very thick] node[pos=0.2, fill=white, inner sep=2pt] {\scriptsize $1$}  node[pos=0.76, fill=white, inner sep=2pt] {\scriptsize $2$} (y3);
       
       \node[vertex, label=0:$c_1$] (c1) at (-1, 3.5) {};
       \node[vertex, label=0:$c_3$] (c3) at (1, 3.5) {};
       \node[vertex, label=0:$c_5$] (c5) at (3, 3.5) {};
       
       \draw (w1) edge[dashed, gray, very thick] node[pos=0.2, fill=white, inner sep=2pt] {\scriptsize $3$}  node[pos=0.76, fill=white, inner sep=2pt] {\scriptsize $1$} (c1); 
       \draw (w2) edge[dashed, gray, very thick] node[pos=0.2, fill=white, inner sep=2pt] {\scriptsize $3$}  node[pos=0.76, fill=white, inner sep=2pt] {\scriptsize $1$} (c3);
       \draw (w3) edge[dashed, gray, very thick] node[pos=0.2, fill=white, inner sep=2pt] {\scriptsize $3$}  node[pos=0.76, fill=white, inner sep=2pt] {\scriptsize $1$} (c5); 
       
     \end{tikzpicture}

   \end{center}
   \caption{An example of a variable gadget for the variable $x_i$ occurring, where $x_i$ occurs exactly in the clauses $C_1$, $C_3$, and~$C_5$. Free edges are marked by wavy lines, while interconnecting edges are dashed and gray.}\label{fvariablegadget}
 \end{figure}
 
 \bfseries Claim: \mdseries $\mathcal{I} $ is a YES-instance if and only if $\mathcal{I}'$ admits a strongly/super-stable matching.
 
 ($\Rightarrow$) Let $T$ be a satisfying truth assignment such that for each clause, exactly one literal is true.
 For each true variable~$x_i$ in this assignment, let $M$ contain the edges $w_i^\ell c_{i_\ell}$ and $y^\ell_i z^\ell_i$ for each $\ell \in [3]$.
 For all other variables, let $M$ contain $w^\ell_i y^\ell_i$ for each $\ell \in [3]$.
 For each clause $C_i$, add the edge $a_i b_i$ to $M$.
 
 Following these rules, we have constructed a matching. It remains to check that $M$ is super-stable (and thus also strongly stable).
 Since $a_i$ is matched to its only neighbor, it cannot be part of a blocking edge.
 Since each $c_i$ is matched along an interconnecting edge, which is better than $b_i$, no blocking edge involves~$b_i$.
 A blocking interconnecting edge $c_i w^\ell_j$ implies that $w^\ell_j$ is not matched to any $y^\ell_j$, however this is only true if $c_i w^\ell_j\in M$.
 A blocking edge $w^\ell_i y_i^j$ does not appear. Either $w^\ell_i$ is matched to its unique first choice $y^\ell_i$ and therefore not part of a blocking edge, or $y^j_i$ is matched to its unique first choice $z^j_i$, and thus, $y^j_i$ is not part of a blocking edge.
 
 ($\Leftarrow$) Let $M$ be a strongly stable matching (note that any super-stable matching is also strongly-stable).
 Then $M$ contains the edge $a_i b_i$, and $c_i$ is matched to a vertex $w_j^\ell$ for all $i\in [m]$, as else $c_i b_i$ or $a_i b_i$ block~$M$.
 If $w_j^\ell c_i\in M$, then $y_j^a z_j^a\in M$ for all $a\in [3]$, as else $w_j^\ell y_j^a$ would be a blocking edge.
 This, however, implies that $w_j^a c_{j_a}\in M$ for all $a\in [3]$, as else $w_j^a c_{j_a}$ would be a blocking edge.
 
 Thus, for each variable $x_i$, the matching $M$ contains either all edges $w_i^\ell c_{i_\ell}$ for $\ell\in [3]$ or none of these edges.
 Thus, the variables $x_i$ such that $M$ contains $w_i^\ell c_{i_\ell}$ for $\ell\in [3]$ induce a truth assignment such that for each clause, exactly one literal is true.
\end{proof}

This proof aimed at the hardness of the restricted case, in which the underlying graph has a low maximum degree. For the sake of completeness, we add another variant, which is defined in a complete bipartite graph.

\begin{proposition}
\label{thm:free_complete}
 \textsc{ssmti-free} is $\NP$-complete, even in complete bipartite graphs, where each tie has length at most three.
\end{proposition}

\begin{proof}
 We reduce from \textsc{ssmti-free}.
 Given a \textsc{ssmti-free} instance on graph $G$, we add all non-present edges between men and women as free edges, ranked worse than any edge from~$E(G)$. We call the resulting graph~$H$.

 Clearly, a strongly/super-stable matching in $G$ is also strong\-ly/super-stable in $H$, as we only added free edges.
 
 Vice versa, let $M$ be a strongly/super-stable matching in $H$.
 Let $M':= M\cap E(G)$ arise from $M$ by deleting all edges not in~$E(G)$.
 Then $M'$ clearly is a matching in $G$, so it remains to show that $M'$ is strongly/super-stable.
 
 Assume that there is a blocking edge $u w$ in, $G$, in the strongly/super-stable sense.
 Since~$u w$ is not blocking in $H$, at least one of $u$ and $w$ has to be matched in $H$, but not in~$G$.
 However, this vertex prefers $u w$ also to its partner in $H$, and thus, $u w$ is also blocking in $H$, which is a contradiction.
\end{proof}

Note that \textsc{ssmti-free} becomes polynomial-time solvable if only a constant number of edges is free in the same way as \textsc{max-ssmi}, the problem of finding  a maximum-cardi\-na\-li\-ty stable matching with strict lists and free edges~\cite{AIKMP13}.

\begin{proposition}
\label{thm:par_free}
 \textsf{SSMTI-Free} can be solved in $\mathcal{O}(2^k n m)$ time in the strongly stable case, and in $\mathcal{O}(2^k m)$ time in the super-stable case, where $k:= |F|$ is the number of free edges.
\end{proposition}

\begin{proof}
For each subset $Q\subseteq F$ of free edges, we construct an instance of \textsc{ssmti-forced} as follows.
Mark all edges in $Q$ as forced, and delete all other free edges.

If any of the \textsc{ssmti-forced} instances admits a stable matching, then this is clearly a stable matching in the \textsc{ssmti-free} instance, as only free edges were deleted.
Vice versa, any solution $M$ for the \textsc{ssmti-free} instance containing exactly the set of forced edges $Q$ (i.e.\ $Q= M\cap F$) immediately implies a solution for the \textsc{ssmti-forced} instance with forced edges~$Q$.

Clearly, there are $2^{k}$ subsets of $F$.
Since any instance of \textsc{ssmti-forced} can be solved in $\mathcal{O} (nm)$ time in the strongly stable case~\cite{FIM07} and in $\mathcal{O}(m)$ time in the super-stable case~\cite{Kun18}, the running time follows.
\end{proof}

\section{Conclusion}
Studying the stable marriage problem with ties combined with restricted edges,
we have shown three $\NP$-completeness results. Our computational hardness results naturally lead to the question whether imposing master lists on both sides makes the problems easier to solve. Moreover, it is open whether \textsc{smt-forbidden-1} remains hard in bounded-degree graphs. In addition, one may try to identify relevant parameters for our problems and then decide whether they are fixed-parameter tractable or admit a polynomial-sized kernel with respect to these parameters.

\paragraph{Acknowledgments.} The authors thank David Manlove and Rolf Niedermeier for useful suggestions that improved the presentation of this paper.

\bibliographystyle{abbrv}
\bibliography{mybib}

\begin{thebibliography}{10}

\bibitem{AIKMP13}
G.~Askalidis, N.~Immorlica, A.~Kwanashie, D.~F. Manlove, and E.~Pountourakis.
\newblock Socially stable matchings in the {H}ospitals / {R}esidents problem.
\newblock In F.~Dehne, R.~Solis-Oba, and J.-R. Sack, editors, {\em Algorithms
  and Data Structures}, volume 8037 of {\em Lecture Notes in Computer Science},
  pages 85--96. Springer Berlin Heidelberg, 2013.

\bibitem{Bir17}
P.~Bir{\'o}.
\newblock Applications of matching models under preferences.
\newblock {\em Trends in Computational Social Choice}, page 345, 2017.

\bibitem{CF09}
K.~Cechl\'{a}rov\'{a} and T.~Fleiner.
\newblock Stable roommates with free edges.
\newblock Technical Report 2009-01, Egerv\'{a}ry Research Group on
  Combinatorial Optimization, Operations Research Department, E\"otv\"os
  Lor\'and University, 2009.

\bibitem{CM16}
{\'A}.~Cseh and D.~F. Manlove.
\newblock Stable marriage and roommates problems with restricted edges:
  complexity and approximability.
\newblock {\em Discrete Optimization}, 20:62--89, 2016.

\bibitem{DFFS03}
V.~M.~F. Dias, G.~D. da~Fonseca, C.~M.~H. de~Figueiredo, and J.~L. Szwarcfiter.
\newblock The stable marriage problem with restricted pairs.
\newblock {\em Theoretical Computer Science}, 306:391--405, 2003.

\bibitem{FIM07}
T.~Fleiner, R.~W. Irving, and D.~F. Manlove.
\newblock Efficient algorithms for generalised stable marriage and roommates
  problems.
\newblock {\em Theoretical Computer Science}, 381:162--176, 2007.

\bibitem{FIM11}
T.~Fleiner, R.~W. Irving, and D.~F. Manlove.
\newblock An algorithm for a super-stable roommates problem.
\newblock {\em Theoretical Computer Science}, 412(50):7059--7065, 2011.

\bibitem{GS62}
D.~Gale and L.~S. Shapley.
\newblock College admissions and the stability of marriage.
\newblock {\em American Mathematical Monthly}, 69:9--15, 1962.

\bibitem{GJ79}
M.~R. Garey and D.~S. Johnson.
\newblock {\em Computers and Intractability}.
\newblock Freeman, San Francisco, CA., 1979.

\bibitem{Irv94}
R.~W. Irving.
\newblock Stable marriage and indifference.
\newblock {\em Discrete Applied Mathematics}, 48:261--272, 1994.

\bibitem{IMO09}
R.~W. Irving, D.~F. Manlove, and G.~O'Malley.
\newblock Stable marriage with ties and bounded length preference lists.
\newblock {\em Journal of Discrete Algorithms}, 7(2):213--219, 2009.

\bibitem{IMS03}
R.~W. Irving, D.~F. Manlove, and S.~Scott.
\newblock Strong stability in the {H}ospitals / {R}esidents problem.
\newblock In {\em Proceedings of STACS '03: the 20th Annual Symposium on
  Theoretical Aspects of Computer Science}, volume 2607 of {\em Lecture Notes
  in Computer Science}, pages 439--450. Springer, 2003.

\bibitem{IMMM99}
K.~Iwama, D.~F. Manlove, S.~Miyazaki, and Y.~Morita.
\newblock Stable marriage with incomplete lists and ties.
\newblock In J.~Wiedermann, P.~van Emde~Boas, and M.~Nielsen, editors, {\em
  Proceedings of ICALP '99: the 26th International Colloquium on Automata,
  Languages, and Programming}, volume 1644 of {\em Lecture Notes in Computer
  Science}, pages 443--452. Springer, 1999.

\bibitem{Knu76}
D.~Knuth.
\newblock {\em Mariages Stables}.
\newblock Les Presses de L'Universit\'{e} de Montr\'{e}al, 1976.
\newblock {E}nglish translation in \emph{Stable Marriage and its Relation to
  Other Combinatorial Problems}, volume 10 of CRM Proceedings and Lecture
  Notes, American Mathematical Society, 1997.

\bibitem{Kun18}
A.~Kunysz.
\newblock An algorithm for the maximum weight strongly stable matching problem.
\newblock In {\em 29th International Symposium on Algorithms and Computation
  (ISAAC 2018)}, pages 42:1--42:13. Schloss Dagstuhl-Leibniz-Zentrum fuer
  Informatik, 2018.

\bibitem{Kwa15}
A.~Kwanashie.
\newblock {\em Efficient algorithms for optimal matching problems under
  preferences}.
\newblock PhD thesis, University of Glasgow, 2015.

\bibitem{Man99}
D.~F. Manlove.
\newblock Stable marriage with ties and unacceptable partners.
\newblock Technical Report TR-1999-29, University of Glasgow, Department of
  Computing Science, January 1999.

\bibitem{Man02}
D.~F. Manlove.
\newblock The structure of stable marriage with indifference.
\newblock {\em Discrete Applied Mathematics}, 122(1-3):167--181, 2002.

\bibitem{Man13}
D.~F. Manlove.
\newblock {\em Algorithmics of Matching Under Preferences}.
\newblock World Scientific, 2013.

\bibitem{MIIMM02}
D.~F. Manlove, R.~W. Irving, K.~Iwama, S.~Miyazaki, and Y.~Morita.
\newblock Hard variants of stable marriage.
\newblock {\em Theoretical Computer Science}, 276:261--279, 2002.

\bibitem{MM10}
E.~McDermid and D.~F. Manlove.
\newblock Keeping partners together: algorithmic results for the {H}ospitals /
  {R}esidents problem with couples.
\newblock {\em Journal of Combinatorial Optimization}, 19:279--303, 2010.

\bibitem{PSSW14}
S.~Porschen, T.~Schmidt, E.~Speckenmeyer, and A.~Wotzlaw.
\newblock {XSAT} and {NAE-SAT} of linear {CNF} classes.
\newblock {\em Discrete Applied Mathematics}, 167:1--14, 2014.

\bibitem{Sch78}
T.~J. Schaefer.
\newblock The complexity of satisfiability problems.
\newblock In {\em Proceedings of the Tenth Annual ACM Symposium on Theory of
  Computing}, pages 216--226. ACM, 1978.

\bibitem{Sco05}
S.~Scott.
\newblock {\em A study of stable marriage problems with ties}.
\newblock PhD thesis, University of Glasgow, Department of Computing Science,
  2005.

\end{thebibliography}
\end{document}